\documentclass[11pt]{article}

\usepackage{amsmath,amssymb,amsthm}
\usepackage{color}
\usepackage[thinlines,thiklines]{easybmat}
\usepackage{physics}

\usepackage{algorithm}
\usepackage{algorithmic}
\usepackage{enumitem}
\setlist{nolistsep}

\usepackage{verbatim}
\usepackage{palatino}
\usepackage{mathpazo}
\usepackage{stmaryrd}
\usepackage[margin=1in]{geometry}
\usepackage{mathtools}
\mathtoolsset{centercolon}

\newtheorem{theorem}{Theorem}
\newtheorem{proposition}{Proposition}

\theoremstyle{definition}
\newtheorem{definition}{Definition}

\usepackage{hyperref}
\definecolor{darkred}  {rgb}{0.5,0,0}
\definecolor{darkblue} {rgb}{0,0,0.5}
\definecolor{darkgreen}{rgb}{0,0.5,0}
\hypersetup{
  pdftitle = {Round Complexity in Key Distillation},
  pdfauthor = {Eric Chitambar},
  colorlinks = true,
  urlcolor  = blue,         
  linkcolor = darkblue,     
  citecolor = darkgreen,    
  filecolor = darkred       
}

\newcommand{\locc}{\overset{\underset{\mathrm{LOCC}}{}}{\longrightarrow}}
\newcommand{\lopc}{\overset{\underset{\mathrm{LOPC}}{}}{\longrightarrow}}

\newcommand{\mc}[1]{\mathcal{#1}}
\newcommand{\mbf}[1]{\mathbf{#1}}
\newcommand{\mbb}[1]{\mathbb{#1}}
\newcommand{\wh}[1]{\widehat{#1}}

\newcommand{\ol}[1]{\overline{#1}}

\title{Round Complexity in the Local Transformations of Quantum and Classical States}

\author{Eric Chitambar $^1$, Min-Hsiu Hsieh $^2$
\\[4mm]
\textit{$^1$ Department of Physics and Astronomy, Southern Illinois University,}\\ 
\textit{Carbondale, Illinois 62901, USA}\\
\textit{$^2$ Centre for Quantum Software and Information (CQSI),}\\
\textit{Faculty of Engineering and Information Technology (FEIT),}\\
\textit{University of Technology Sydney (UTS), NSW 2007, Australia}}

\date{\today}

\begin{document}
\maketitle

\begin{abstract}

A natural operational paradigm for distributed quantum and classical information processing involves local operations coordinated by multiple rounds of public communication.  In this paper we consider the minimum number of communication rounds needed to perform the locality-constrained task of entanglement transformation and the analogous classical task of secrecy manipulation.  Specifically we address whether bipartite mixed entanglement can always be converted into pure entanglement or whether unsecure classical correlations can always be transformed into secret shared randomness using local operations and a \textit{bounded} number of communication exchanges.  Our main contribution in this paper is an explicit construction of quantum and classical state transformations which, for any given $r$, can be achieved using $r$ rounds of classical communication exchanges but no fewer.  Our results reveal that highly complex communication protocols are indeed necessary to fully harness the information-theoretic resources contained in general quantum and classical states.  The major technical contribution of this manuscript lies in proving lower bounds for the required number of communication exchanges using the notion of common information and various lemmas built upon it.  We propose a classical analog to the Schmidt rank of a bipartite quantum state which we call the secrecy rank, and we show that it is a monotone under stochastic local classical operations assisted by iterative classical communication.

\end{abstract}

\section{Introduction}

One of the most fascinating aspects of quantum information is how \textit{classical} communication can enhance \textit{quantum} information processing.  For instance, ``mixed'' entanglement shared between two or more parties can be ``purified'' when the parties are allowed to perform \textbf{L}ocal quantum \textbf{O}perations on their subsystems and \textbf{C}ommunicate \textbf{C}lassically with one another \cite{Bennett-1996a}, a process known as LOCC.  An analog of this purification procedure can be found in the classical theory of secret correlations.  Rather surprisingly, secret correlations shared between two or more parties can be strengthened by the parties performing \textbf{L}ocal stochastic \textbf{O}perations and ``leaking'' information partially through \textbf{P}ublic \textbf{C}ommunication \cite{Bennett-1986a, Maurer-1993a, Ahlswede-1993a}, a process known as LOPC.  In both settings, the principle is the same: resource manipulation (whether it be entanglement or secrecy) becomes more powerful when public classical communication is allowed.

In a communication protocol, the parties take turns exchanging information with one another, information that is extracted locally from their respective subsystems and earlier rounds of communication.  There are different ways of measuring the role that communication plays in a distributed protocol. In the topic of communication complexity, the figure of merit is the number of bits communicated between parties during the protocol.  Classical communication complexity theory was introduced in the seminal paper by Yao \cite{Yao-1979a} and later extended to the quantum setting in Ref.~\cite{Yao-1993a}. On the other hand, the subject of \textbf{round complexity} studies how the number and ordering of interactive communication exchanges affect the parties' ability to perform a particular task, an equally important figure of merit. However despite it being a very natural question to consider, the exact manner and extent to which rounds of communication help resource manipulation is largely unknown. This paper investigates the round complexity of transforming one quantum/classical state to another using LOCC/LOPC.

Previous work on LOCC/LOPC round complexity is relatively limited.  On the classical side, the studies of round complexity have largely been focused on models of interactive function computation \cite{Kaspi-1985a, Beaver-1990a, Orlitsky-1991a, Ma-2011a}.   On the quantum side, clear separations between one-round and two-round protocols have been demonstrated for various quantum information-processing tasks such as asymptotic entanglement distillation \cite{Bennett-1996a}, tripartite entanglement transformations \cite{Chitambar-2011a}, quantum state discrimination  \cite{Xin-2007a, Owari-2008a, Chitambar-2013a, Nathanson-2013a, Chitambar-2013c}, and recently, the simulation of nonlocal gates using shared entanglement \cite{Wakakuwa-2016a}. However, none of these results have been able to establish an operational separation between each of the finite-round LOCC classes in terms of manipulating a single quantum state.  That is, it has been previously unknown whether or not for every finite $r$ there exists an LOCC transformation $\rho\locc\sigma$ that is possible only if $r$ rounds of communication are used.  Such a phenomenon might be unexpected given that every bipartite pure-state transformation $\ket{\psi}\locc\ket{\phi}$ can be accomplished in just one round of LOCC, regardless of the dimensions \cite{Nielsen-1999a}.  In contrast to this elegant result, we show in this paper that $r$ rounds of communication are indeed required to perform certain LOCC transformations.

The operational tasks we explore involve extracting pure-state entanglement from some mixed quantum state using LOCC and the classical analog of extracting secret shared randomness from an unsecure classically-correlated state using LOPC.  These are two very important questions since pure-state entanglement is the fundamental building block for quantum information processing \cite{Horodecki-2009a}, and likewise, secret key states provide the essential ingredient for information-theoretic secure communication \cite{Shannon-1949a, Devetak-2005b}.  Understanding the relationship between quantum entanglement and classical secrecy offers an intriguing research directions with many interesting connections already found \cite{Collins-2002a, Gisin-2002b, Acin-2003b, Horodecki-2005c, Acin-2005b, Christandl-2007a, Oppenheim-2008a, Bae-2009a, Ozols-2014a, Chitambar-2014d}.  This paper describes another similarity between the two in terms of LOCC/LOPC round complexity.

As our main result, we construct families of quantum (resp. classical) states for which a minimum of $r$ communication rounds is both necessary and sufficient to obtain pure-state entanglement (resp. secret shared randomness).  Our findings imply that there exists no universal upper bound on the number of LOCC/LOPC rounds needed to perform such tasks, universal in the sense that it holds for states of all dimensions/alphabet size.  Rigorously proving this claim is a highly non-trivial matter since the general structure of LOCC and LOPC protocols is quite complex, allowing for arbitrary local operations and arbitrary interactive communication schemes \cite{Chitambar-2014b}.  With this complexity, it is difficult to definitively rule out the possibility of some clever round-compression technique that could always reduce the number of communication exchanges below some finite threshold, regardless of the system sizes.  In fact, such a clever round-compression strategy is precisely what allows for the restriction to just one-way protocols for all bipartite pure-state transformations $\ket{\psi}\locc\ket{\phi}$ \cite{Lo-2001a}.  

The paper is structured as follows.  We begin by fixing notation and providing a brief overview of the LOCC and LOPC frameworks.  In Section \ref{Sect:Origami States} we recursively construct a family of tripartite distributions that we call the origami distributions, due to the ``unfolding'' appearance of the construction.  Using these distributions, we define tripartite quantum states through the embedding $p_{xyz}\to\sum_{xyz}\sqrt{p_{xyz}}\ket{xyz}$.  Our main result is contained in Theorem \ref{thm:main}, and its proof is then carried out in Section \ref{Sect:Unified-proofs}.  An essential tool used in the proof is the classical information-theoretic object known as the G\'{a}cs-K\"{o}rner Common Information  \cite{Gacs-1973a}, which we review in Section \ref{Sect:Common-Info}.  Concluding comments and discussion are provided in Section \ref{Sect:Conclusion}.

\subsection*{The LOCC and LOPC Frameworks}


The problems studied in this paper involve two trustworthy parties (Alice and Bob) and one unwanted third party (Eve).  When Alice, Bob, and Eve are holding quantum systems, we denote their joint state by $\rho^{ABE}$.  In contrast when Alice, Bob, and Eve are holding random variables $X$, $Y$, and $Z$, we denote their joint probability distribution by $p^{XYZ}$.  These variables range over sets $\mc{X}$,  $\mc{Y}$, and $\mc{Z}$ respectively, and the probability of event $(x,y,z)$ will be denoted by $p^{XYZ}_{xyz}$.  When the underlying random variables are clear, we will simply write the probabilities by $p_{xyz}$.  Conditional probabilities are denoted, for example, by $p_{xy|z}^{XY|Z=z}$.  

In an LOCC protocol, Alice and Bob take turns performing a local quantum instrument, which is a collection of completely positive (CP) maps $\{\mc{E}_\lambda\}_\lambda$ such that $\sum_\lambda\mc{E}_\lambda$ is trace-preserving \cite{Chitambar-2014b}.  The index $\lambda$ represents the ``measurement outcome'' of the instrument which is communicated to the other party, thereby correlating the choice of future local instruments to previous measurement outcomes.  For the problem considered in this paper, we will be considering instruments in which each local CP maps has the form $\mc{E}_\lambda(\rho)=K_\lambda\rho K_\lambda^\dagger$, where the $\{K_\lambda\}_\lambda$ form a complete set of Kraus operators; i.e. $\sum_\lambda K_\lambda^\dagger K_\lambda=\mbb{I}$.  

In an LOPC protocol, Alice and Bob share random variables $X$ and $Y$ respectively.  They proceed with multiple iterations of public communication where the $i^{th}$ message $M_i$ is the stochastic output of a channel performed to $(P,M_{<i})$, where $P\in\{X,Y\}$ is the variable of the announcing party in the $i^{th}$ round and $M_{<i}=M_1\cdots M_{i-1}$ denotes the sequence of messages generated in the previous $i-1$ rounds.  At the end of the protocol, Alice and Bob generate output variables $\wh{X}$ and $\wh{Y}$ that are obtained by processing $(X,M)$ and $(Y,M)$ respectively, where $M$ represents all communication variables generated throughout the protocol.  For both LOCC and LOPC, an $r$-round protocol consists of $r$ classical communication exchanges between the parties.

One conceptual difference between the LOCC and LOPC settings is that in the latter, the presence of an unwanted eavesdropping party is always taken into account.  Thus, a copy of the public communication $M$ is shared by Eve, and a general LOPC protocol generates a transformation of probability distributions
\begin{equation}
p^{XYZ}\lopc p^{\wh{X}\wh{Y}(ZM)}.
\end{equation}

The fundamental resource unit in entanglement theory is the entangled bit (ebit), which has the form $\ket{\Phi}^{AB}=\sqrt{1/2}(\ket{00}^{AB}+\ket{11}^{AB})$.  In classical secrecy theory, the basic resource unit is the secret bit (sbit).  This is any distribution over the sets $\{0,1\}\times\{0,1\}\times \mc{Z}$ of the form $p^{XYZ}_{xyz}=\frac{1}{2}\delta_{xy}p^Z_z$, where $p^Z$ is an arbitrary distribution for Eve.  Alice and Bob's main concern is how much Eve is correlated with their variables, rather than the specific distribution over her variable.  Hence, we will adopt the notation that $\Phi$ denotes a sbit, with Eve's uncorrelated distribution being unspecified.  For partially entangled two-qubit states and for non-uniform secret shared bits, we will write
\begin{align}
\ket{\Phi_\lambda}&=\sqrt{\lambda}\ket{00}+\sqrt{1-\lambda}\ket{11}\\
\Phi_\lambda&=\lambda \delta_{X0}\delta_{Y0}+(1-\lambda)\delta_{X1}\delta_{Y1}.\label{Eq:partialbit}
\end{align}
Here $\delta_{X0}$, for example, is the distribution over $\mc{X}$ that has $x=0$ with unit probability.  The entropy of $\Phi_\lambda$ is $h(\lambda)$, where $h(x)=-x\log x-(1-x)\log(1-x)$.

As a final bit of notation, for a pure state $\ket{\varphi}^{AB}$, we let $Srk(\ket{\varphi})$ denote the Schmidt rank of the state, which is equivalent to the ranks of the reduced density matrices $\rho^A=\tr_B\op{\varphi}{\varphi}$ and $\rho^B=\tr_A\op{\varphi}{\varphi}$.  In Section \ref{Sect:secrcyRank} we will introduce the notion of \textit{secrecy rank} for tripartite distributions $p^{XYZ}$.  Because of its operational similarity to the quantum Schmidt rank, we will likewise denote this classical quantity by $Srk(p^{XYZ})$.  

\section{Round Complexity in a Family of LOPC and LOCC Transformations}

\label{Sect:Origami States}

In this section we introduce a family of tripartite distributions, which we call the \textit{origami distributions}.  The family is given by the set $\{\mbf{b}^{(i,\lambda)}:i\in\mbb{N},\;0<\lambda\leq 1/2\}$, with $\mbf{b}^{(i,\lambda)}$ being a tripartite probability distribution taking on values $\mbf{b}^{(i,\lambda)}_{xyz}$ for each fixed pair of values $(i,\lambda)$; i.e. $\sum_{xyz}\mbf{b}^{(i,\lambda)}_{xyz}=1$.  The structure of these distributions is described recursively with $\mbf{b}^{(1,\lambda)}$ having the form:
\begin{align}
\label{Eq:B1}
\mbf{b}^{(1,\lambda)}=\begin{BMAT}{ccc}{cc}
{}&\qquad x&{}\\
 \begin{BMAT}{r}{cc}
 \vphantom{\text{\Huge Y}}\\y
 \end{BMAT}&\begin{BMAT}{c|cccc}{c|cccc}
{}&0&1&2&3\\
0&0&\cdot &\cdot&1\\
1&\cdot&0&1&\cdot\\
2&2&3&\cdot&\cdot\\
3&\cdot&\cdot&3&2
\end{BMAT}
& \begin{BMAT}{r}{cccc}
\vphantom{\text{\Huge z}}\\ \vphantom{\text{\Huge z}}\\\vphantom{\text{\Huge z}}\\z
 \end{BMAT}
\end{BMAT}
\end{align}
The 8 events of $(x,y,z)$ having nonzero probability in $\mbf{b}^{(1,\lambda)}$ are those in which $z$ lies in row $y$ and column $x$, as shown in this grid.  The probabilities of these events are $\mbf{b}^{(1,\lambda)}_{xy|z}=\lambda$ for even values of $x$, $\mbf{b}^{(1,\lambda)}_{xy|z}=1-\lambda$ for odd values of $x$, and $\mbf{b}^{(1,\lambda)}_z=1/4$.  In other words, $\mbf{b}^{(1,\lambda)}$ consists of four blocks of uniform probability, each corresponding to a different value of $z$.  Within each block $X$ and $Y$ are perfectly correlated but with a non-uniform distribution $(\lambda,1-\lambda)$.

For each fixed value of $\lambda$, we now proceed to build the $i^{th}$ distribution in the family $\{\mbf{b}^{(i,\lambda)}:i\in\mbb{N},\;0<\lambda\leq 1/2\}$ according to the following prescription:
\begin{align}
\begin{cases} \text{Define:}\quad \ol{\mbf{b}^{(1,\lambda)}}=\begin{BMAT}{ccc}{cc}
{}&\qquad x&{}\\
 \begin{BMAT}{r}{cc}
 \vphantom{\text{\Huge Y}}\\y
 \end{BMAT}&\begin{BMAT}{c|cccc}{c|cccc}
{}&0&1&2&3\\
0&4&\cdot &\cdot&5\\
1&\cdot&\cdot&7&6\\
2&6&7&\cdot&\cdot\\
3&\cdot&4&5&\cdot
\end{BMAT}
& \begin{BMAT}{r}{cccc}
\vphantom{\text{\Huge z}}\\ \vphantom{\text{\Huge z}}\\\vphantom{\text{\Huge z}}\\z
 \end{BMAT}
\end{BMAT},\\
{}\\
\text{Even $n$:} \quad \mbf{b}^{(n,\lambda)}=\begin{bmatrix}\mbf{b}^{(n-1,\lambda)}&\ol{\mbf{b}^{(n-1,\lambda)}}\end{bmatrix}\qquad (\text{size: $\;2^{n/2+2}\times 2^{n/2+1}\times 2^{n+1}$)},\\
{}\\
\text{Odd $n$:} \quad \mbf{b}^{(n,\lambda)}=\begin{bmatrix}\mbf{b}^{(n-1,\lambda)}\\\ol{\mbf{b}^{(n-1,\lambda)}}\end{bmatrix}\qquad (\text{size: $\;2^{(n-1)/2+2}\times 2^{(n-1)/2+2}\times 2^{n+1}$)},
\end{cases}
\end{align}
where $\ol{\mbf{b}^{(n,\lambda)}}$ is obtained from $\mbf{b}^{(n,\lambda)}$ by interchanging the row (resp. column) $i$ with row (resp. column) $i+2^{\lfloor n/2+1\rfloor}$ for all odd $i$ whenever $n$ is odd (resp. even), and Eve's values are increased by $2^n$ from the original values in $\mbf{b}^{(n,\lambda)}$.  In each grid, all of Eve's values are still equiprobable, and for each value of $z$, Alice and Bob have shared randomness with $\mbf{b}^{(n,\lambda)}_{xy|z}=\lambda$ for even values of $x$.  For example,
\begin{align}
\mbf{b}^{(2,\lambda)}&=\begin{BMAT}{ccc}{cc}
{}&\qquad x&{}\\
 \begin{BMAT}{r}{cc}
 \vphantom{\text{\Huge Y}}\\y
 \end{BMAT}&\begin{BMAT}{c|cccc:cccc}{c|cccc}
{}&0&1&2&3&4&5&6&7\\
0&0&\cdot&\cdot&1&4&\cdot&\cdot&5\\
1&\cdot&0&1&\cdot&\cdot&\cdot&7&6\\
2&2&3&\cdot&\cdot&6&7&\cdot&\cdot\\
3&\cdot&\cdot&3&2&\cdot&4&5&\cdot
\end{BMAT}
& \begin{BMAT}{r}{cccc}
\vphantom{\text{\Huge z}}\\ \vphantom{\text{\Huge z}}\\\vphantom{\text{\Huge z}}\\z
 \end{BMAT}
\end{BMAT}
&\mbf{b}^{(3,\lambda)}&=\begin{BMAT}{ccc}{cc}
{}&\qquad x&{}\\
 \begin{BMAT}{r}{cc}
 \vphantom{\text{\Huge Y}}\\y
 \end{BMAT}&\begin{BMAT}{c|cccccccc}{c|cccc:cccc}
{}&0&1&2&3&4&5&6&7\\
0&0&\cdot&\cdot&1&4&\cdot&\cdot&5\\
1&\cdot&0&1&\cdot&\cdot&\cdot&7&6\\
2&2&3&\cdot&\cdot&6&7&\cdot&\cdot\\
3&\cdot&\cdot&3&2&\cdot&4&5&\cdot\\
4&8&\cdot&\cdot&13&12&\cdot&\cdot&9\\
5&\cdot&\cdot&9&14&\cdot&8&15&\cdot\\
6&10&15&\cdot&\cdot&14&11&\cdot&\cdot\\
7&\cdot&12&11&\cdot&\cdot&\cdot&13&10
\end{BMAT}
& \begin{BMAT}{r}{ccccccc}
\vphantom{\text{\Huge z}}\\\vphantom{\text{\Huge z}}\\ \vphantom{\text{\Huge z}}\\\vphantom{\text{\Huge z}}\\ \vphantom{\text{\Huge z}}\\\vphantom{\text{\Huge z}}\\z
 \end{BMAT}
\end{BMAT}
\end{align}

\begin{align}
\mbf{b}^{(4,\lambda)}&=\begin{BMAT}{ccc}{cc}
{}&\qquad x&{}\\
 \begin{BMAT}{r}{cc}
 \vphantom{\text{\Huge Y}}\\y
 \end{BMAT}&\begin{BMAT}{c|cccccccc:cccccccc}{c|cccccccc}
{}&0&1&2&3&4&5&6&7&8&9&10&11&12&13&14&15\\
0&0&\cdot&\cdot&1&4&\cdot&\cdot&5&16&\cdot&\cdot&17&20&\cdot&\cdot&21\\
1&\cdot&0&1&\cdot&\cdot&\cdot&7&6&\cdot&\cdot&25&30&\cdot&24&31&\cdot\\
2&2&3&\cdot&\cdot&6&7&\cdot&\cdot&18&19&\cdot&\cdot&22&23&\cdot&\cdot\\
3&\cdot&\cdot&3&2&\cdot&4&5&\cdot&\cdot&28&27&\cdot&\cdot&\cdot&29&26\\
4&8&\cdot&\cdot&13&12&\cdot&\cdot&9&24&\cdot&\cdot&29&28&\cdot&\cdot&25\\
5&\cdot&\cdot&9&14&\cdot&8&15&\cdot&\cdot&16&17&\cdot&\cdot&\cdot&23&22\\
6&10&15&\cdot&\cdot&6&11&\cdot&\cdot&26&31&\cdot&\cdot&30&27&\cdot&\cdot\\
7&\cdot&12&11&\cdot&\cdot&\cdot&13&10&\cdot&\cdot&19&18&\cdot&20&21&\cdot
\end{BMAT}
& \begin{BMAT}{r}{ccccccc}
\vphantom{\text{\Huge z}}\\\vphantom{\text{\Huge z}}\\ \vphantom{\text{\Huge z}}\\\vphantom{\text{\Huge z}}\\ \vphantom{\text{\Huge z}}\\\vphantom{\text{\Huge z}}\\z.
 \end{BMAT}
\end{BMAT}
\end{align}


We now use the origami distributions to construct bipartite quantum states.  This is accomplished by first embedding each distribution $\mbf{b}^{(i,\lambda)}$ into a tripartite quantum state according to
\begin{align}
\ket{\mbf{b}^{(i,\lambda)}}^{ABE}&=\sum_{x,y,z}\sqrt{\mbf{b}^{(i,\lambda)}_{xyz}}\ket{x}^A\ket{y}^B\ket{z}^E\notag\\
&=\frac{1}{\sqrt{2^{i+1}}}\sum_z\ket{\psi_z^{(i,\lambda)}}^{AB}\ket{z}^E,\label{Eq:origami-quantum}
\end{align}
where $\ket{\psi_z^{(i,\lambda)}}^{AB}:=\sum_{x,y}\sqrt{\mbf{b}^{(i,\lambda)}_{xy|z}}\ket{x}^A\ket{y}^B$.  Notice that the von Neumann entropy of this state is $h(\lambda)$ for every $z$ and $i$.  Alice and Bob's reduced state is then given by $\rho_{\mbf{b}}^{(i,\lambda)}:=\tr_E\left(|\mbf{b}^{(i,\lambda)}\rangle\langle\mbf{b}^{(i,\lambda)}|\right)$.  The main results of this paper are stated in the following theorem.
\begin{theorem}
\label{thm:main}
For any pair $(r,\lambda)$ and any $0<\lambda'\leq 1/2$, the LOPC transformation 
\begin{align}
\mbf{b}^{(r,\lambda)}&\lopc \Phi_{\lambda'}
\end{align}
and the LOCC transformation
\begin{align}
\rho_{\mbf{b}}^{(r,\lambda)}&\locc\op{\Phi_{\lambda'}}{\Phi_{\lambda'}}
\end{align}
are both impossible using $r-1$ rounds of communication exchanges; nor are they possible in $r$ rounds if Alice (resp.~Bob) is the first to announce when $r$ is odd (resp.~even).  Conversely, for $\lambda'\leq \lambda\leq 1/2$ the transformations are possible in $r$ rounds if Bob (resp.~Alice) is the first to announce when $r$ is odd (resp.~even).
\end{theorem}

In Section \ref{Sect:Unified-proofs} we prove this theorem using a unified argument that applies to both the quantum and classical problems.  The main idea is that in either a quantum or classical protocol, completing the desired transformation requires that the rank of the state remain invariant at each step of the protocol.  The notion of ``rank'' here refers to the Schmidt rank in the quantum case, and the secrecy rank in the classical case, a quantity we introduce in Section \ref{Sect:secrcyRank}.  The origami distributions are designed in such a way that obtaining the desired target state in $r-1$ rounds along one branch will necessarily cause the rank of the state to decrease along another branch.  The ranks can only be preserved along all branches if the protocol is carried out for $r$ total rounds.  In order to reach these conclusions, we need a way to compactly represent and analyze the structure of the origami distributions for arbitrary $r$.  To this end we turn to the notion of \textit{common information} between random variables, in the sense proposed by G\'{a}cs and K\"{o}rner.  We first briefly review some general properties of the common information, and then we apply it to the origami distributions.

\section{G\'{a}cs-K\"{o}rner Common Information}

\label{Sect:Common-Info}
 
For a general pair of random variables $AB$ with distribution $p^{AB}$, there exists a \textit{maximal} common variable $J_{AB}$ in the sense that $J_{AB}$ can be computed exactly from either $A$ or $B$, and any other such common function of $A$ and $B$ is itself a function of $J_{AB}$.  Hence, up to relabeling, the variable $J_{AB}$ is unique for each pair of variables $AB$, and G\'{a}cs and K\"{o}rner identify $H(J_{AB})$ as the common information of $AB$ \cite{Gacs-1973a}.  For values $a,a'\in\mc{A}$, it is not difficult to show that $J_{AB}(a)=J_{AB}(a')$ iff there exists a sequence of values
\begin{equation}
\label{Eq:communicating}
ab_1a_1b_2a_2\cdots a_n a'
\end{equation}
with $a,a_1,\cdots, a_n,a'\in\mc{A}$ and $b_1,\cdots b_{n}\in\mc{B}$ such that $p_{ab_1}p_{b_1a_1}p_{a_1b_2}\cdots p_{b_na'}>0$ \cite{Gacs-1973a, Chitambar-2014c}.  


One can go further and introduce the \textit{maximal conditional common function} \cite{Chitambar-2014c, Chitambar-2014d}.  For three random variables $ABC$, a maximal conditional common function $J_{AB|C}$ is the collection of variables $\{J_{AB|C=c}:c\in\mc{C}\}$ with $J_{AB|C=c}$ being a maximal common function of the conditional distribution $p^{AB|C=c}$.  The variable $J_{XY|Z}$ is again unique for every distribution $p_{XYZ}$ up to relabeling.  For all distributions considered in this paper, including the origami distributions, we will assume that some canonical ordering has be fixed (and known to all parties) so that we may speak unambiguously of \textit{the} maximal common function $J_{AB}$ and \textit{the} maximal conditional common function $J_{AB|C}$.

Let us now analyze the origami distributions in terms of the G\'{a}cs-K\"{o}rner common information.  We first focus on random variables $XYZ$ whose distribution is given by $\mbf{b}^{(1,\lambda)}$.  From Eq. \eqref{Eq:communicating} and the graphical representations of \eqref{Eq:B1}, we can see that $x$ and $x'$ satisfy $J_{XZ}(x)=J_{XZ}(x')$ iff a path connects the columns corresponding to $x$ and $x'$ such that movement (possibly diagonal) from one column occurs only through a common value of $z$.  A similar rule stipulates that $y$ and $y'$ satisfies $J_{YZ}(y)=J_{YZ}(y')$ iff the rows corresponding to $y$ and $y'$ are connected by a path that only switches rows if a common value of $z$ belongs to both rows.  In Eq.~\eqref{Eq:B1}, we see that $J_{XZ}$ is constant (trivial) while $J_{YZ}$ is binary outcome with its value determined by whether $y\in\{0,1\}$ or $y\in\{2,3\}$.  In a similar way, $J_{XY}(x)=J_{XY}(x')$ iff the columns of $x$ and $x'$ can be connected by a path that changes columns iff those columns have possible events occurring in the same row.  Hence, $J_{XY}$ is trivial for the distribution depicted in Eq.~\eqref{Eq:B1}.

The origami distributions are constructed precisely to satisfy the following proposition, which can be proven by inspection and using simple inductive arguments following the discussion of the previous paragraph.
\begin{proposition}
\label{Prop:DistProp}
For any fixed value of $\lambda$, let $X^{(n)}Y^{(n)}Z^{(n)}$ denote random variables whose distribution is given by $\mbf{b}^{(n,\lambda)}$.  Then
\begin{enumerate}
\item The variable $J_{X^{(n)}Y^{(n)}}$ is trivial for all $n$;
\item For odd (resp.~even) $n$, the variable $J_{X^{(n)}Z^{(n)}}$ (resp. $J_{Y^{(n)}Z^{(n)}}$) is trivial while $J_{Y^{(n)}Z^{(n)}}$ (resp.~$J_{X^{(n)}Z^{(n)}}$) is binary;
\item For odd (resp.~even) $n$, the distribution $\mbf{b}^{(n,\lambda)}$ is equivalent (up to relabeling) to $\mbf{b}^{(n-1,\lambda)}$ when conditioned on the value of $J_{Y^{(n)}Z^{(n)}}$ (resp. $J_{X^{(n)}Z^{(n)}}$).
\item $I(X^{(n)}:Y^{(n)}|Z^{(n)})=H(J_{X^{(n)}Y^{(n)}|Z^{(n)}}|Z^{(n)})=h(\lambda)$ for all $n$. 
\end{enumerate}
\end{proposition}

\section{The Proof of Theorem \ref{thm:main}}

\label{Sect:Unified-proofs}

\subsection{Achievability}

First consider the classical case.  Given Proposition \ref{Prop:DistProp}, it is easy to see that the transformation $\mbf{b}^{(r,\lambda)}\lopc\Phi_\lambda$ is possible in $r$ rounds: each party alternates in announcing his/her common information with Eve, with Alice (resp. Bob) going first when $r$ is even (resp. odd).  With the state $\Phi_\lambda$, the transformation to $\Phi_{\lambda'}$ can always be performed whenever $\lambda'<\lambda$ \cite{Collins-2002a}.  Such a transformation requires one-way communication, but this communication can always be included in the $r^{th}$ round message of the protocol.  The $r$-round achievability in the quantum case is equivalent to the classical protocol with Alice and Bob replacing their common information announcement with the corresponding two-outcome projective measurements. 

\subsection{Necessity}

Both transformations are clearly impossible when $\lambda'>\lambda$, which can be seen by appealing to monotonicity of LOCC/LOPC monotones.  When $\lambda'>\lambda$ the so-called entanglement of formation would need to been increased in the LOCC transformation (which is not possible \cite{Bennett-1996a}), and the analgous conditional mutual information $I(X:Y|Z)$ would need to be increased in the LOPC transformation (which is likewise not possible \cite{Csiszar-2011a}).  Henceforth, we restrict attention to the case that $\lambda\geq\lambda'$.  The proof is separated into quantum and classical parts.

\subsubsection{The Quantum Scenario}

\label{Sect:proof-quantum}

Let us begin by introducing some new notation based on the block diagrams of $\mbf{b}^{(r,\lambda)}$.  Let $\mc{B}^{(r)}$ be the set of events $(x,y,z)$ such that $\mbf{b}^{(r,\lambda)}_{xyz}>0$.  For every $k=1,\cdots,r-1$, there exists a disjoint partitioning of $\mc{B}^{(r)}$ into subsets $\mc{B}_{j_1,\cdots,j_k}^{(r-k)}$ such that $(x,y,z)\in\mc{B}_{j_1,\cdots,j_k}^{(r-k)}$ if $\mbf{b}^{(r,\lambda)}_{xyz|j_1,\cdots,j_k}>0$, where $j_i\in\{0,1\}$ is the value of Alice's (resp. Bob's) common information with Eve given all previous values $j_1,\cdots,j_{i-1}$ when $r-(i-1)$ is even (resp. odd).  In other words, the sets $\mc{B}_{j_1,\cdots,j_k}^{(r-k)}$ are the supports of the different sub-distributions $\mbf{b}^{(r-k,\lambda)}$ used to build $\mbf{b}^{(r,\lambda)}$ in the recursive construction.  With a slight abuse of terminology, we will write, for instance, $x\in\mc{B}_{j_1,\cdots,j_k}^{(r-k)}$ if there exists some $(y,z)$ such that $(x,y,z)\in\mc{B}_{j_1,\cdots,j_k}^{(r-k)}$.  Note that for every $z$ and fixed $k$, there exists one and only one set $\mc{B}_{j_1,\cdots,j_k}^{(r-k)}$ such that $z\in\mc{B}_{j_1,\cdots,j_k}^{(r-k)}$.

Recall that $\rho^{(r,\lambda)}_{\mbf{b}}$ has the decomposition $\left\{\frac{1}{2^{r+1}}, \ket{\psi_z^{(r,\lambda)}}\right\}$. We will drop the superscript in the state $\ket{\psi_z^{(r,\lambda)}}$ in the proof to ease notation. The deterministic transformation $\rho^{(r,\lambda)}_{\mbf{b}}\locc\op{\Phi_{\lambda'}}{\Phi_{\lambda'}}$ requires that the LOCC protocol transforms $\ket{\psi_z}\to\ket{\Phi_{\lambda'}}$ for every $\ket{\psi_z}$.  Since both the initial and final states have Schmidt rank two in the transformation $\ket{\psi_z}\to\ket{\Phi_{\lambda'}}$, every local operation must either eliminate the state $\ket{\psi_z}$ or preserve its Schmidt rank.  We will use this crucial fact to first argue that at the end of $r-1$ rounds on $\rho^{(r,\lambda)}_{\mbf{b}}$ there will always be some outcome branch for which four entangled states have not been eliminated, each of them having values associated with the same block $\mc{B}_{j_1,\cdots,j_{r-1}}^{(1)}$.  We then show that it is impossible for all of these four entangled states to be simultaneously transformed into $\ket{\Phi_{\lambda'}}$ using local operations and \textit{no} communication.  Therefore, there will always be some outcome branch in which an entangled pure state is not obtained.

We proceed with the following inductive argument whose validity when $k=0$ is trivial.

\noindent \textbf{Inductive assumption:}  Along some branch at the end of round $k$, there exists a set $\mc{B}_{j_1,\cdots,j_k}^{(r-k)}$ such that $\ket{\psi_z}$ has not been eliminated for all $z\in\mc{B}_{j_1,\cdots,j_k}^{(r-k)}$.  Let us denote the posterior states of $\ket{\psi_z}$ at this point in the protocol by $\ket{\psi_z'}$ (i.e. $\ket{\psi_z'}\propto \; A\otimes B\ket{\psi_z}$, where $A\otimes B$ is the full Kraus operator representing all local measurements performed up to this point of the protocol).  Thus, $\ket{\psi_z'}$ is a rank-two entangled state for all $z\in\mc{B}_{j_1,\cdots,j_k}^{(r-k)}$.  We denote $\ket{x'}\propto\;A\ket{x}$ and $\ket{y'}\propto\;B\ket{y}$ as the posterior states of $\ket{x}$ and $\ket{y}$ at this point in the protocol for all $(x,y)\in\mc{B}_{j_1,\cdots,j_k}^{(r-k)}$.  Note that the $\ket{x'}$ and $\ket{y'}$ span the local supports of the $\ket{\psi_z'}$.  

Without loss of generality, suppose that $r-k$ is even.  We consider first the case when Bob is the measuring party in round $k+1$.  Since $r-k$ is even, Bob and Eve do not share any common information in the distribution $\mbf{b}^{(r-k,\lambda)}$ (see Proposition \ref{Prop:DistProp}).  Since $Srk(\ket{\psi_z'})=Srk(\ket{\psi_z})$ for all $z\in\mc{B}_{j_1,\cdots,j_k}^{(r-k)}$, an application of Proposition \ref{Prop:QuantumCI} (see below) shows that for any $\ol{y},\ol{z}\in\mc{B}_{j_1,\cdots,j_k}^{(r-k)}$ there exists a sequence
\[\ol{y}z_1,y_1,z_2,y_2\cdots,z_ny_n\ol{z}\] 
with $y_i,z_i\in\mc{B}_{j_1,\cdots,j_k}^{(r-k)}$ such that
\begin{equation}
\label{Eq:CI-link}
\bra{\ol{y}'}\rho^{\prime B}_{z_1}\ket{\ol{y}'}\bra{y_1'}\rho^{\prime B}_{z_1}\ket{y_1'}\bra{y_1'}\rho^{\prime B}_{z_2}\ket{y_1'}\cdots\bra{y'_n}\rho^{\prime B}_{\ol{z}}\ket{y_n'}>0,
\end{equation}
where $\rho^{\prime B}_{z_i}=\tr_A\op{\psi_{z_i}'}{\psi_{z_i}'}$.  Here we are using the inductive assumption that every $\ket{\psi_z'}$ is a rank-two entangled state for $z\in\mc{B}_{j_1,\cdots,j_k}^{(r-k)}$; since if, say, $\ket{y_1}\in supp[\rho^B_{z_1}]$, then $\ket{y'_1}\in supp[\rho_{z_1}^{\prime B}]$ when $\ket{\psi'_z}$ is entangled.

Now let $\{B_\mu\}_{\mu}$ be Kraus operators characterizing Bob's local measurement in round $k+1$.  We argue that for every value of $\mu$, either $\mbb{I}\otimes B_\mu\ket{\psi_z'}=0$ for all $z\in\mc{B}_{j_1,\cdots,j_k}^{(r-k)}$, or $\mbb{I}\otimes B_\mu\ket{\psi_z'}\not=0$ for all $z\in\mc{B}_{j_1,\cdots,j_k}^{(r-k)}$.  Indeed, suppose there is some $z$ for which $\mbb{I}\otimes B_\mu\ket{\psi_{z}'}=0$.  This means $B_\mu\ket{\ol{y}'}=0$ for $\ket{\ol{y}'}$ in the support of $\rho_{z}^{\prime B}$.  Consider now any $\ol{z}\in \mc{B}_{j_1,\cdots,j_k}^{(r-k)}$.  There must exist some sequence $\ol{y}z_1,y_1,z_2,y_2\cdots,z_ny_n\ol{z}$ such that Eq. \eqref{Eq:CI-link} holds.  The fact that $\bra{\ol{y}'}\rho^{\prime B}_{z_1}\ket{\ol{y}'}>0$ (as seen from Eq. \eqref{Eq:CI-link}) implies that $B_\mu\rho^{\prime B}_{z_1} B_\mu^\dagger$ no longer has rank two since $B_\mu\ket{\ol{y}'}=0$.  However, every $\mbb{I}\otimes B_\mu\ket{\psi_z'}$ either must be eliminated or have rank two; hence $B_\mu\rho^{\prime B}_{z_1} B_\mu^\dagger$ must vanish.  From the second term in Eq. \eqref{Eq:CI-link}, we see that $\bra{y_1'}\rho^{\prime B}_{z_1}\ket{y_1'}>0$, which means that $B_\mu\ket{y_1'}=0$ in order for $B_\mu\rho^{\prime B}_{z_1} B_\mu^\dagger$ to vanish.  Continuing along the sequence of Eq. \eqref{Eq:CI-link} and repeating this argument, we will eventually reach $\rho_{\ol{z}}^{\prime B}$, which also must be eliminated by $B_\mu$.  Since $\ol{z}$ was arbitrary, we have established that $\mbb{I}\otimes B_\mu\ket{\psi_z'}=0$ for some $z\in\mc{B}_{j_1,\cdots,j_k}^{(r-k)}$ implies that $\mbb{I}\otimes B_\mu\ket{\psi_z'}=0$ for all $z\in\mc{B}_{j_1,\cdots,j_k}^{(r-k)}$.  Finally, since $\{B_\mu\}_{\mu}$ is a complete measurement, there must exist at least one outcome $\ol{\mu}$ such that $\mbb{I}\otimes B_{\ol{\mu}}\ket{\psi_{\ol{z}}'}\not=0$ for some $\ol{z}\in \mc{B}_{j_1,\cdots,j_k}^{(r-k)}$.  Let $j_{k+1}\in\{0,1\}$ be such that $\ol{z}\in\mc{B}_{j_1,\cdots,j_k,j_{k+1}}^{(r-k-1)}$.  Since $\mc{B}_{j_1,\cdots,j_k,j_{k+1}}^{(r-k-1)}$ is just a subset of $\mc{B}_{j_1,\cdots,j_k}^{(r-k)}$, we have that $\mbb{I}\otimes B_\mu\ket{\psi_z'}\not=0$ for all $z\in\mc{B}_{j_1,\cdots,j_k,j_{k+1}}^{(r-k-1)}$.  This verifies the inductive assumption when Bob is measuring in round $k+1$.

We now consider the case when Alice is the measuring party in round $k+1$ with $r-k$ being even as before.  In this case, Alice and Eve share one bit of common information in the distribution $\mbf{b}^{(r-k,\lambda)}$.  However, this information simply specifies whether a given $z\in\mc{B}_{j_1,\cdots,j_k}^{(r-k)}$ either belongs to $\mc{B}_{j_1,\cdots,j_k,0}^{(r-k-1)}$ or $\mc{B}_{j_1,\cdots,j_k,1}^{(r-k-1)}$.  Most importantly, Alice has no common information with Eve for values $(x,z)$ within either $\mc{B}_{j_1,\cdots,j_k,0}^{(r-k-1)}$ or $\mc{B}_{j_1,\cdots,j_k,1}^{(r-k-1)}$.  Hence by repeating the previous argument within those sub-blocks, we have that for her measurement $\{A_\mu\}_{\mu}$ either $A_\mu\otimes\mbb{I}\ket{\psi_z'}=0$ for all $z\in\mc{B}_{j_1,\cdots,j_k,0}^{(r-k-1)}$, or $A_\mu\otimes \mbb{I}\ket{\psi_z'}\not=0$ for all $z\in\mc{B}_{j_1,\cdots,j_k,0}^{(r-k-1)}$; and likewise $A_\mu\otimes\mbb{I}\ket{\psi_z'}=0$ for all $z\in\mc{B}_{j_1,\cdots,j_k,1}^{(r-k-1)}$, or $A_\mu\otimes \mbb{I}\ket{\psi_z'}\not=0$ for all $z\in\mc{B}_{j_1,\cdots,j_k,1}^{(r-k-1)}$.  There must be at least one outcome $\ol{\mu}$ with $A_{\ol{\mu}}\otimes\mbb{I}\ket{\psi_{\ol{z}}'}\not=0$ for some $\ol{z}\in\mc{B}_{j_1,\cdots,j_k}^{(r-k)}$.  Let $j_{k+1}\in\{0,1\}$ be such that $\ol{z}\in\mc{B}_{j_1,\cdots,j_k,j_{k+1}}^{(r-k-1)}$.  This verifies the inductive assumption for round $k+1$.

Having proven the inductive assertion, we now apply it to a hypothetical $(r-1)$-round LOCC protocol that transforming $\rho_{\mbf{b}}^{(r,\lambda)}$ into $\ket{\Phi_{\lambda'}}$.  At the end of this protocol, there must exist some set $\mc{B}_{j_1,\cdots,j_{r-1}}^{(1)}$ such that the posterior probability of the state $\ket{\psi_z}$ is nonzero for every $\mc{B}_{j_1,\cdots,j_{r-1}}^{(1)}$.   But since there are no more rounds left in the protocol, each of these states must be locally convertible into the target state $\ket{\Phi_{\lambda'}}$ with no further communication.  Proposition \ref{Prop:QuantumNoCC} below shows that this is not possible, and therefore the hypothetical $(r-1)$-round protocol performing the desired transformation does not exist.

The final part of the proof is to show that the transformation is impossible in $r$ rounds if Alice (resp. Bob) is the first to announce when $r$ is odd (resp. even).  The reasoning follows exactly along the lines of the proceeding argument.  Consider the case when Alice is announcing first and $r$ is odd.  In this case she initially shares no common information with Eve.  Therefore, her measurement is unable to eliminate any of the $\ket{\psi_z}$ and so at the end of round $1$, all values of $z$ still belong to $\mc{B}^{(r)}$.  We can then repeat the about argument except with the modified inductive assumption: Along some branch at the end of round $k$ there exists a set $\mc{B}^{(r-(k-1))}_{j_1,\cdots j_{k-1}}$ such that $\ket{\psi_z}$ has not been eliminated for all $z\in\mc{B}^{(r-(k-1))}_{j_1,\cdots j_{k-1}}$.  Proceeding for $r$ rounds again leaves Alice and Bob with at least one block $\mc{B}^{(1)}_{j_1,\cdots j_{r-1}}$ with no states eliminated.  From Proposition \ref{Prop:QuantumNoCC}, the transformation cannot be completed.

We now prove the two main propositions referenced in the above proof.  The first essentially says that when embedding a probability distribution into a tripartite quantum states, common information cannot be generated between Eve and any one of the other two parties, even with just a nonzero probability. 
\begin{proposition}
\label{Prop:QuantumCI}
Let $\{\ket{\psi_{z}}^{AB}\}_{\mc{Z}}$ be a collection of bipartite states for Alice and Bob's systems and $Z$ a random variable ranging over $\mc{Z}$ with distribution $p^Z$.  Suppose that $H(J_{XZ})=H(J_{YZ})=0$ for distributions
\begin{align}
p^{XZ}_{xz}&=\bra{x}(\tr_B\op{\psi_{z}}{\psi_{z}})\ket{x}&&\text{and}&p^{YZ}_{yz}&=\bra{y}(\tr_A\op{\psi_{z}}{\psi_{z}})\ket{y}.
\end{align}
Let $A$ and $B$ be any pair of operators such that $Srk(\ket{\psi_z'})=Srk(\ket{\psi_z})$ for all $z$, where $Srk(\cdot)$ is the Schmidt rank of the given state and $\ket{\psi_z'}=\frac{A\otimes B\ket{\psi_z}}{\sqrt{\bra{\psi_z}A^\dagger A\otimes B^\dagger B\ket{\psi_z}}}$.  Then for any $\ol{x},\ol{z}$ there exists a sequence $(x_i,z_i)_i$ such that
\begin{equation}
\label{Eq:probCIquantum1}
\bra{\ol{x}'}\rho^{\prime A}_{z_1}\ket{\ol{x}'}\bra{x_1'}\rho^{\prime A}_{z_1}\ket{x_1'}\bra{x'_1}\rho^{\prime A}_{z_2}\ket{x'_1}\cdots\bra{x'_n}\rho^{\prime A}_{\ol{z}}\ket{x'_n}>0,
\end{equation}
where $\rho^{\prime A}_i=\tr_{B}\op{\psi_z'}{\psi_z'}$ and $\ket{x'}=\frac{\ket{x}}{\sqrt{\bra{x}A^\dagger A\ket{x}}}$.  An analogous statement holds for any pair of values $\ol{y},\ol{z}$.
\end{proposition}
\begin{proof}
By the assumption that $H(J_{XZ})=0$, we have that for any $\ol{x},\ol{z}$ there exists a sequence $(x_i,z_i)_i$ such that
\begin{equation}
p^{XZ}_{\ol{x},z_1}p^{XZ}_{x_1,z_1}p^{XZ}_{x_1,z_2}\cdots p^{XZ}_{x_n,z_n}p^{XZ}_{x_n,\ol{z}}>0.
\end{equation}
The essential observation is that if $Srk(A\otimes B\ket{\psi_z})=Srk(\ket{\psi_z})$, then
\begin{equation}
\bra{x}(\tr_B\op{\psi_{z}}{\psi_{z}})\ket{x}>0\qquad\Rightarrow\qquad \bra{x}A^\dagger (\tr_B A\otimes B\op{\psi_{z}}{\psi_{z}}A^\dagger \otimes B^\dagger )A\ket{x}>0.
\end{equation}
Indeed, the first inequality says that (i) $\ket{x}\in supp \left(\tr_B\op{\psi_{z}}{\psi_{z}}\right)$.  Since $A\otimes B$ does not decrease the rank of $\ket{\psi_z}$, this means that (ii) $supp \left(\tr_B\op{\psi_{z}}{\psi_{z}}\right)=supp \left(\tr_B\mbb{I}\otimes B\op{\psi_{z}}{\psi_{z}}\mbb{I}\otimes B^\dagger\right)$, and (iii) $A\ket{x}\not=0$.  Combining facts (i)--(iii) gives that $A\ket{x}\in supp\left(\tr_B A\otimes B\op{\psi_{z}}{\psi_{z}}A^\dagger \otimes B^\dagger \right)$.  By interchanging $x$ and $y$ in this argument, an analogous inequality to \eqref{Eq:probCIquantum1} is proven for any pair $\ol{y},\ol{z}$.
\end{proof}

The second proposition provides the final ingredient in the above proof.  Up to relabeling, the events in $\mc{B}_{j_1,\cdots,j_{r-1}}^{(1)}$ corresponds to the events in the support of $\mbf{b}^{(1,\lambda)}$.  For reference, we reproduce the diagram depicting $\mbf{b}^{(1,\lambda)}$:
\begin{align}
\mbf{b}^{(1,\lambda)}=\begin{BMAT}{ccc}{cc}
{}&\qquad x&{}\\
 \begin{BMAT}{r}{cc}
 \vphantom{\text{\Huge Y}}\\y
 \end{BMAT}&\begin{BMAT}{c|cccc}{c|cccc}
{}&0&1&2&3\\
0&0&\cdot &\cdot&1\\
1&\cdot&0&1&\cdot\\
2&2&3&\cdot&\cdot\\
3&\cdot&\cdot&3&2
\end{BMAT}
& \begin{BMAT}{r}{cccc}
\vphantom{\text{\Huge z}}\\ \vphantom{\text{\Huge z}}\\\vphantom{\text{\Huge z}}\\z
 \end{BMAT}
\end{BMAT}
\end{align}
\begin{proposition}
\label{Prop:QuantumNoCC}
Let $\{\ket{\psi_z}\}_{z=0}^3$ be the four entangled states obtained by the embedding of $\mbf{b}^{(1,\lambda)}$.  Let $\{\ket{\psi_z'}\}_{z=0}^3$ be the resulting four states at the end of one branch in an LOCC protocol, all of them being rank-two entangled.  Then it is not possible to transform each of the states into $\ket{\Phi_{\lambda'}}$ using just local operations with no further communication.
\end{proposition}

\begin{proof}

Let $A\otimes B$ be the measurement operators corresponding to this branch in the protocol.  In other words $\ket{\psi_z'}\propto \;A\otimes B\ket{\psi_z}$, and we can write
\begin{align}
\ket{\psi'_0}&\propto\left(\ket{0'}^{A}\ket{0'}^B+\ket{1'}^A\ket{1'}^B\right)\label{Eq:B11}\\
\ket{\psi'_1}&\propto\left(\ket{2'}^{A}\ket{0'}^B+\ket{3'}^A\ket{1'}^B\right)\label{Eq:B12}\\
\ket{\psi'_2}&\propto\left(\ket{0'}^{A}\ket{2'}^B+\ket{3'}^A\ket{3'}^B\right)\label{Eq:B13}\\
\ket{\psi'_3}&\propto\left(\ket{1'}^{A}\ket{2'}^B+\ket{2'}^A\ket{3'}^B\right)\label{Eq:B14}
\end{align}
where $\ket{i'}^A\propto\;A\ket{i}^A$ and  $\ket{j'}^B\propto\;B\ket{j}^B$ for all $j$.  The local operation will consist in each party performing a local measurement and unitary rotation.  Since no communication is allowed, the target state $\ket{\Phi_{\lambda'}}$ must be obtained after every possible combination of outcomes.  Let $B_0$ be one of Bob's measurement operators that doesn't eliminate either $\ket{\psi_0'}^B$ or $\ket{\psi_1'}^B$, and let $B_1$ be one of Bob's measurement operators that doesn't eliminate either $\ket{\psi_2'}^B$ or $\ket{\psi_3'}^B$ (such operators must exist).  Its clear that by applying $B_0$ to the first pair and $B_1$ to the second pair, the form of Eqns. \eqref{Eq:B11}--\eqref{Eq:B14} will stay the same.  Likewise, there must exist some operator $A_0$ of Alice that does not eliminate either $\ket{0'}^A$ or $\ket{1'}^A$.  Additionally, this operator cannot eliminate either $\ket{3'}^A$ or $\ket{2'}^A$, or else the ranks of $\ket{\psi_2'}$ and $\ket{\psi_3'}$ will drop to one respectively.  Thus, applying operator $A_0$ to the states will not change their form either.

Without loss of generality then, we can assume that each of the $\ket{\psi_z'}$ are nonzero, proportional to $\ket{\Phi_{\lambda'}}$, and therefore proportional to each other.  The linear independence of $\ket{0'}^B$ and $\ket{1'}^B$ in  Eqns. \eqref{Eq:B11} and \eqref{Eq:B12} implies that $\ket{2'}^A\propto\ket{0'}^A$ and $\ket{3'}^A\propto\ket{1'}^A$.  Eq. \eqref{Eq:B11} and \eqref{Eq:B12} then give $\ket{2'}^B\propto\ket{0'}^B$ and $\ket{3'}^B\propto\ket{1'}^B$.  However, this forces Eq. \eqref{Eq:B14} to have the form $\alpha\ket{1'}^A\ket{0'}^B+\beta\ket{0'}^A\ket{1'}^B$ which contradicts the fact that it is proportional to 
$\ket{0'}^{A}\ket{0'}^B+\ket{1'}^A\ket{1'}^B$.

\end{proof}

\subsubsection{The Classical Scenario}

In this section we will prove that the corresponding LOPC transformation $\mbf{b}^{(r,\lambda)}\lopc\Phi_{\lambda'}$ is not possible in $r-1$ rounds.   For notational clarity, we let $X_nY_nZ_n$ denote random variables that are jointly distributed according to $\mbf{b}^{(n,\lambda)}$.  Our goal here is to reproduce the proof of Section. \ref{Sect:proof-quantum} in terms of an LOPC transformation of $\mbf{b}^{(r,\lambda)}$.  In the proof above, we made heavy use of the Schmidt rank of a bipartite pure state.  What is the analog of the Schmidt rank for classical distributions?  We propose one such quantity in Section \ref{Sect:secrcyRank}, which we call the \textit{secrecy rank}.  

For the simple structure of the origami distributions, the secrecy rank of the conditional distribution $p^{X_nY_n|Z_n=z}$ is equivalent to the number of events having nonzero probability given $Z_n=z$.  Since $\mbf{b}^{(r,\lambda)}\lopc\Phi_{\lambda'}$ is a deterministic transformation, then for every $z$, the conditional distribution $p^{X_rY_r|Z_r=z}$ must be transformed into $\Phi_{\lambda'}$ with probability one.  Therefore, with both $p^{X_rY_r|Z_r=z}$ and $\Phi_{\lambda'}$ having secrecy rank two, monotonicity of the secrecy rank (Theorem \ref{thm:rankmonotone} below) implies that every local operation must either ``eliminate'' the distribution $p^{X_rY_r|Z_r=z}$ or preserve its secrecy rank.  In other words, for every sequence of messages $m_{\leq k}=(m_1,\cdots,m_k)$ either $p_{z|m_{\leq k}}=0$ or both events having nonzero events in the original distribution $p^{X_rY_r|Z_r=z}$ still have nonzero posterior probability when given messages $m_{\leq k}$.  This is analogous to the Schmidt rank condition we have in the quantum case.

With this connection of rank preservation established, we can now run the exact same inductive as in the quantum proof of Section \ref{Sect:proof-quantum}. Hence, we can conclude that in any LOPC protocol transforming $\mbf{b}^{(r,\lambda)}\lopc\Phi_{\lambda'}$: Along some branch at the end of every round $k$, there exists a set $\mc{B}_{j_1,\cdots,j_k}^{(r-k)}$ such that $p^{X_rY_r|Z_r=z}$ has not been eliminated for all $z\in\mc{B}_{j_1,\cdots,j_k}^{(r-k)}$.  Consequently, after $r-1$ rounds, there will be some set of events $\mc{B}_{j_1,\cdots,j_{r-1}}^{(1)}$ corresponding to the block structure of $\mbf{b}^{(1,\lambda)}$, except with the events having possibly different nonzero posterior probabilities than at the start of the protocol.  The crucial point, however, is that no events in this set have been eliminated at the end of the $r-1$ rounds.  It is very easy to see that Alice and Bob then share no common information along this branch of the protocol, and therefore the only perfectly correlated variable that they can agree on is a trivial one.  Indeed, recall that non-trivial common information exists if the events with nonzero probability form disjoint blocks in the distribution; from Eq. \eqref{Eq:B1} we see this is not possible when all events have a nonzero probability.  Thus, we have proven that $r-1$ rounds of LOPC are not sufficient to complete the transformation $\rho^{(r,\lambda)}_{\mbf{b}}\locc\op{\Phi_{\lambda'}}{\Phi_{\lambda'}}$.  When Alice (resp. Bob) is the first to announce for $r$ is odd (resp. even), impossibility of the transformation in $r$ rounds can be argued just as in the quantum case above.

$\qed$


\subsection{A Classical Analog to the Schmidt rank}

\label{Sect:secrcyRank}

In this section we propose an LOPC analog to the quantum Schmidt rank which we will call the secrecy rank.  The construction is based on the so-called \textit{secret key cost} of a tripartite distribution \cite{Renner-2003a}, a quantity whose single-letter characterization \cite{Winter-2005a, Chitambar-2014e} has close connections to Wyner's classic notion of common information \cite{Wyner-1975a}.  A detailed exploration of the relationship between all these quantities is beyond the scope of this paper and will be saved for future work.

In what follows, for a distribution $p^{W}$ over the set $\mc{W}$, we let $|p^W|$ denote the number of events in   $\mc{W}$ with a nonzero probability.
\begin{definition}
The \textbf{secrecy rank} of tripartite distribution $p^{XYZ}$ is defined as
\begin{equation}
\label{Eq:DefnSecrecyRank}
Srk[p^{XYZ}]=\min_{X-ZW-Y}\max_z|p^{W|Z=z}|,
\end{equation}
where the minimization is taken over all auxiliary random variables $W$ such that $I(X:Y|ZW)=0$.  
\end{definition}

Let us describe how this quantity is analogous to the quantum Schmidt rank.  First consider the case when $Z$ is trivial; i.e. $|p^Z|=1$.  The Schmidt decomposition of a bipartite pure state has the form $\ket{\varphi}^{AB}=\sum_{w=1}^{Srk(\ket{\varphi})}\sqrt{p_w}\ket{\alpha_w}^A\ket{\beta_w}^B$ where the $\{\ket{\alpha_w}^A\}$ and $\{\ket{\beta_w}^B\}$ form orthonormal bases for Alice and Bob's systems respectively.  Suppose that Alice and Bob both measure $\ket{\varphi}^{AB}$ by projecting in their Schmidt basis.  If $X$ (resp. $Y$) is the random variable describing Alice's (resp. Bob's) outcomes, then their measurement statistics can be described as the marginal of the tripartite distribution $p^{XYW}$ where $p^{XYW}_{xyw}=\delta_{xw}\delta_{yw}p^W_w$; i.e. $X-W-Y$.  Clearly $Srk(\ket{\varphi})=Srk(p^{XY})$.  In the case that $Z$ is not trivial, the definition of Eq. \eqref{Eq:DefnSecrecyRank} most closely resembles the definition of Schmidt rank for bipartite mixed states, as proposed in Ref. \cite{Terhal-2000a}.  Namely, for a density matrix $\rho^{AB}$, one minimizes the quantity $Srk(\mathfrak{E})$ over all pure-state ensembles $\mathfrak{E}=\{\ket{\psi_i}^{AB},q_i\}$ generating $\rho^{AB}$, where $Srk(\mathfrak{E})$ is the maximum Schmidt rank of all the states in $\mathfrak{E}$.  For classical distributions $p^{XYZ}$, one can think of $p^{XYZ}$ as defining an ensemble of bipartite classical states $\mathfrak{E}=\{p^{XY|Z=z}, p^Z_z\}$.  There is no minimization over ensembles as in the quantum case, and therefore one obtains the secrecy rank of $p^{XYZ}$ by just taking the maximum secrecy rank of all the states in $\mathfrak{E}$.  This is precisely what Eq. \eqref{Eq:DefnSecrecyRank} gives.

We now prove a crucial operational property of the secrecy rank.  If one of the parties, say Alice, locally generates a message $M$, then the resulting distribution is $p^{XYZM}$ with secrecy rank given by 
\begin{align}
Srk[p^{XY(ZM)}]=\min_{XM-ZMW-YM}\max_{z,m}|p^{W|Z=z,M=m}|.
\end{align}

For distribution $p^{XYZ}$, let $W_0$ be any variables such that $Srk[p^{XYZ}]=\max_{z}|p^{W_0|Z=z}|$, and suppose that Alice generates a public message $M$.  This implies that $MX-W_0Z-Y\Rightarrow X-MW_0Z-Y\Rightarrow MX-MW_0Z-MY$.  Hence, the secrecy rank of the distribution after Alice's message satisfies
\begin{align}
\label{Eq:SrkMonotone}
Srk[p^{XY(ZM)}]&=\min_{XM-ZMW-YM}\max_{z,m}|p^{W|Z=z,M=m}|\notag\\
&=\min_{X-ZMW-YM}\max_{z,m}|p^{W|Z=z,M=m}|\notag\\
&\leq \max_{z,m}|p^{W_0|Z=z,M=m}|\notag\\
&\leq \max_{z}|p^{W_0|Z=z}|=Srk[p^{XYZ}].
\end{align}
This chain of inequalities shows that the Schmidt rank is an LOPC monotone.  However, an even stronger statement can be made since $\max_{z,m}|p^{W|Z=z,M=m}|\geq \max_{z}|p^{W|Z=z,M=m'}|$ for any fixed message $m'$.  Hence from Eq. \eqref{Eq:SrkMonotone} we can conclude that
\begin{equation}
Srk[p^{XYZ}]\geq Srk[p^{XY(ZM)}|M=m']\quad\forall m'\in\mc{M}.
\end{equation}
This shows that the secrecy rank cannot be increased even when conditioned on just a single message; in other words, with zero probability can the secrecy rank be increased by LOPC.  Quantities having this property are known as \textit{stochastic} LOPC (SLOPC) monotones.  This result is completely analogous to the quantum setting in which the Schmidt rank is known to be a stochastic LOCC (SLOCC) monotone.  In summary,
\begin{theorem}
\label{thm:rankmonotone}
The secrecy rank is a \textit{stochastic} LOPC (SLOPC) monotone. 
\end{theorem}

\section{Discussion and Conclusion}

\label{Sect:Conclusion}

In this paper we have investigated the question of round complexity in the local transformation of classical and quantum states.  We have shown that no universal upper bound exists on the minimum number of rounds needed to perform a a general bipartite entanglement transformation as well as a bipartite extraction of secret shared randomness from unsecure correlations.  We close this paper with some additional observations and open questions.

First, it should be emphasized that the LOCC impossibiltiy result of Theorem \ref{thm:main} holds for \textit{any} $0<\lambda'\leq\lambda$.  In particular, the target state $\ket{\Phi_{\lambda'}}$ can be entangled by an arbitrarily small amount and the transformation still requires $r$ rounds.  This demonstrates a type of discontinuity in the trade-off between entanglement and LOCC round number since $\ket{00}=\lim_{\lambda'\to 0}\ket{\Phi_{\lambda'}}$ can be trivially obtained in zero rounds of LOCC.  Such a phenomenon is reminiscent of the entanglement/round number trade-off demonstrated in Ref. \cite{Chitambar-2011a}.

It is also noteworthy that the classical notion of common information played an essential role in our line of argumentation.  Being able to unify the classical and quantum problems in this manner required the origami distributions to have special structure.  Ozols \textit{et al.} have previously used distributions of this sort to relate the tasks of classical and quantum key distillation \cite{Ozols-2014a}, and it appears that distributions with this structure provide a useful starting point for investigating the similarities and differences between quantum entanglement and classical secrecy theories. 

In addition, we remark that that the bipartite quantum states $\rho^{(r,\lambda)}_\mbf{b}$ constructed in this paper exhibit entanglement reversibility in the asymptotic sense \cite{Horodecki-1998b, Vollbrecht-2004a, Cornelio-2011a, Horodecki-2009a}.  That is, the entanglement cost of generating $\rho^{(r,\lambda)}_{\mbf{b}}$ by LOCC is equal to the amount of entanglement that can be distilled from $\rho^{(r,\lambda)}_{\mbf{b}}$, which is $h(\lambda)$.  A very simple protocol for generating $\rho^{(r,\lambda)}_{\mbf{b}}$ at entanglement rate $h(\lambda)$ involves Alice and Bob converting $Nh(\lambda)$ copies of $\ket{\Phi_{1/2}}$ into $N$ copies of $\ket{\Phi_\lambda}$.  On each of these copies Alice and Bob then choose a random joint permutation consistent with the block structure of $\mbf{b}^{(r,\lambda)}$: $\ket{\Phi_\lambda}\to\ket{\psi_z^{(r,\lambda)}}$.  Averaging over these permutation generates the state $\rho^{(r,\lambda)}_{\mbf{b}}$.  Our results show that the general structure of states possessing entanglement reversibility can be highly complex.  It is an interesting question of whether the entanglement distillation rate of $h(\lambda)$ can still be achieved for $\rho^{(r,\lambda)}_{\mbf{b}}$ in the asymptotic sense using fewer than $r$ rounds of LOCC.  We strongly conjecture that this is not possible, but we offer no definitive proof.

Another natural question to consider is the greatest success probability for achieving the transformations $\rho_\mbf{b}^{(r,\lambda)}\to\ket{\Phi_\lambda}$ and $\mbf{b}^{(r,\lambda)}\to\Phi_\lambda$ using $r-1$ rounds of LOCC and LOPC respectively.  The structure of the $\mbf{b}^{(r,\lambda)}$ suggests that in both cases the success probability of any $(r-1)$-round protocol should be no greater than $1/2$.  In fact, it is not difficult to construct an $(r-1)$-round protocol that exactly attains the success probability $1/2$.  We can prove that this indeed is optimal for the classical case, but we are no longer able to easily map this bound to the quantum setting like we have done in this paper.  The main reason is that monotonicity of the Schmidt/secrecy rank is no longer required in the transformation.  Therefore, the unified analysis of the quantum and classical scenarios pursued in this paper no longer holds.  We suspect that an LOCC/LOPC equivalence can still be established by using tools other than the Schmidt/secrecy rank.  This is left for future work.

\bigskip 

\noindent\textit{\textbf{Acknowledgments}}

We thank Benjamin Fortescue for fruitful discussions on multi-round key distillation.  EC is supported by the National Science Foundation (NSF) Early CAREER Award No. 1352326. MH is supported by an ARC Future Fellowship under Grant FT140100574.

\bibliographystyle{alphaurl}
\bibliography{QuantumBib}

\end{document}